\documentclass[12pt,draftcls,onecolumn]{IEEEtran}
\usepackage{amsfonts}
\usepackage{amsmath,amssymb}
\usepackage{graphicx}
\usepackage{caption2}
\usepackage{epsfig}

\def\dref#1{(\ref{#1})}
\def\rm{\mathrm}
\newtheorem{theorem}{Theorem}
\newtheorem{lemma}{Lemma}

\newtheorem{remark}{Remark}

\begin{document}

\title{Consensus of Multi-Agent Systems
with General Linear and Lipschitz Nonlinear Dynamics Using Distributed Adaptive Protocols}
\author{ Zhongkui~Li, Wei~Ren,~\IEEEmembership{Member,~IEEE}, Xiangdong~Liu, and Mengyin~Fu%
\thanks{Z. Li, X. Liu and M. Fu are with the School of
Automation, Beijing Institute of Technology, Beijing 100081, China (E-mail: zhongkli@gmail.com).}
\thanks{W. Ren is with the Department of
Electrical and Computer Engineering, Utah State University, UT
84322, USA (E-mail: wei.ren@usu.edu).}}

\maketitle

\begin{abstract} This paper considers the distributed
consensus problems for multi-agent systems with general linear and
Lipschitz nonlinear dynamics.
Distributed relative-state consensus protocols with an adaptive
law for adjusting the coupling weights between neighboring agents
are designed for both the linear and nonlinear cases,
under which consensus is reached for all
undirected connected communication graphs. Extensions
to the case with a leader-follower communication graph are further studied.
In contrast to the existing results in the literature, the adaptive
consensus protocols here can be implemented by
each agent in a fully distributed fashion without using any global information.
\end{abstract}

\section{Introduction}

In recent years, the consensus problem for multi-agent systems has
received compelling attention from various scientific communities,
for its potential applications in such broad areas as spacecraft
formation flying, sensor networks, and cooperative surveillance
\cite{olfati-saber2007consensus,ren2007information}. A general
framework of the consensus problem for networks of integrator agents
with fixed and switching topologies is addressed in
\cite{olfati-saber2004consensus}. The conditions given by
\cite{olfati-saber2004consensus} are further relaxed in
\cite{ren2005consensus}. A distributed algorithm is proposed in
\cite{cortes2008distributed} to achieve consensus in finite time.
Distributed $H_\infty$ consensus and control problems are
investigated in \cite{lin2008distributed,li2009h} for networks of
agents subject to external disturbances and model uncertainties.
Consensus algorithms are designed in
\cite{carli2009quantized,li2010distributed} for a group of agents
with quantized communication links and limited data rate. The authors in
\cite{rahmani2008controllability} studies the controllability of
leader-follower multi-agent systems from a graph-theoretic
perspective. To ensure that the states of a group of agents follow a
reference trajectory of a leader, consensus tracking algorithms are
given in \cite{hong2008distributed,ren2007multi-vehicle} for agents
with fixed and switching topologies. A passivity-based design
framework is proposed in \cite{arcak2007passivity} to achieve group
coordination. The consensus problems for networks of double- and
high-order integrators are studied in
\cite{ren2008consensus,ren2007high-order,jiang2009consensus}.
Readers are referred to the recent surveys
\cite{olfati-saber2007consensus,ren2007information} for a relatively
complete coverage of the literature on consensus.

This paper considers the distributed consensus problems for multi-agent systems
with general linear and Lipschitz nonlinear dynamics.
Consensus of multi-agent systems with general linear dynamics was previously studied in
\cite{li2010consensus,seo2009consensus,tuna2009conditions,scardovi2009synchronization}.
In particular, different static and dynamic consensus protocols are
designed in \cite{li2010consensus,seo2009consensus}, requiring the
smallest nonzero eigenvalue of the Laplacian matrix associated with the
communication graph to be known by each agent to determine the bound for
the coupling weight. However, the Laplacian matrix depends on the entire communication graph
and is hence global
information. In other
words, these consensus protocols in
\cite{li2010consensus,seo2009consensus} can not be computed and
implemented by each agent in a fully distributed fashion, i.e.,
using only local information of its own and neighbors. To tackle
this problem, we propose here a distributed consensus protocol based
on the relative states combined with an adaptive law for adjusting the
coupling weights between neighboring agents, which is partly
inspired by the edge-based adaptive strategy for the synchronization
of complex networks in
\cite{delellis2009novel,delellis2010synchronization}.

The proposed distributed adaptive protocols are designed, respectively,
for linear and Lipschitz nonlinear multi-agent systems, under which
consensus is reached in both the linear and the nonlinear cases
for any undirected connected communication graph.
It is shown that a sufficient condition for the
existence of such a protocol in the linear case
is that each agent is stabilizable. Existence conditions for the adaptive protocol
in the nonlinear case are also discussed. It is pointed out that
the results in the nonlinear case can be reduced to those in the linear
case, when the Lipschitz nonlinearity does not exist. Extensions
of the obtained results to the case with a leader-follower
communication graph are further discussed. It is
worth noting that the consensus protocols in
\cite{tuna2009conditions,scardovi2009synchronization} can also
achieve consensus for all connected communication graphs. Contrary
to the general linear and Lipschitz nonlinear agent dynamics in this paper,
the linear agent dynamics in \cite{tuna2009conditions} are restricted to be
neutrally stable and all the eigenvalues of
the state matrix of each agent in \cite{scardovi2009synchronization}
are assumed to lie in the closed left-half plane. 
In addition,
adaptive synchronization of a class of complex network satisfying a
Lipschitz-type condition is considered in \cite{delellis2009novel,delellis2010synchronization}.
However, the results given in \cite{delellis2009novel,delellis2010synchronization} require the
inner coupling matrix to be positive semi-definite, which is not
directly applicable to the consensus problem under
investigation here.

The rest of this paper is organized as follows. The adaptive
consensus problems for multi-agent systems with general linear and Lipschitz
nonlinear dynamics are considered, respectively, in Sections II and III.
Extensions to the case with a leader-follower communication graph
are studied in Section IV. Simulation examples are presented
to illustrate the analytical results in Section V. Conclusions are
drawn in Section VI.

Throughout this paper, the following notations will be used:
Let $\mathbf{R}^{n\times n}$
be the set of $n\times n$ real matrices. The superscript $T$ means
transpose for real matrices. $I_N$ represents the identity matrix of
dimension $N$. Matrices, if not explicitly stated, are assumed to
have compatible dimensions. Denote by $\mathbf{1}$ the column vector
with all entries equal to one. $\rm{diag}(A_1,\cdots,A_n)$
represents a block-diagonal matrix with matrices $A_i,i=1,\cdots,n,$
on its diagonal. For real symmetric matrices $X$
and $Y$, $X>(\geq)Y$ means that $X-Y$ is positive (semi-)definite.
$A\otimes B$ denotes the Kronecker product of matrices $A$ and $B$.

\section{Adaptive Consensus for Multi-Agent Systems with General
Linear Dynamics}

Consider a group of $N$ identical agents with general
linear dynamics. The dynamics of the $i$-th agent
are described by
\begin{equation}\label{1c}
\begin{aligned}
    \dot{x}_i &=Ax_i+Bu_i,\quad i=1,\cdots,N,
\end{aligned}
\end{equation}
where $x_i\in\mathbf{R}^n$ is the state,
$u_i\in\mathbf{R}^{p}$ is the control input,
and $A$, $B$, are constant matrices with
compatible dimensions.

The communication topology among the agents is represented by an
undirected graph $\mathcal {G}=(\mathcal {V},\mathcal {E})$, where
$\mathcal {V}=\{1,\cdots,N\}$ is the set of nodes (i.e., agents),
and $\mathcal {E}\subset\mathcal {V}\times\mathcal {V}$ is the set
of edges. An edge $(i,j)$ $(i\neq j)$ means that agents $i$ and $j$
can obtain information from each other. A path in $\mathcal {G}$
from node $i_1$ to node $i_l$ is a sequence of edges of the
form $(i_k, i_{k+1})$, $k=1,\cdots,l-1$. An undirected graph is
connected if there exists a path between every pair of distinct
nodes, otherwise is disconnected.

A variety of static and dynamic consensus protocols have been
proposed to reach consensus for agents with dynamics given by \dref{1c}, e.g., in
\cite{li2010consensus,seo2009consensus,tuna2009conditions,scardovi2009synchronization}.
For instance, a static consensus protocol based on the relative states
between neighboring agents is given in \cite{li2010consensus} as
\begin{equation}\label{clc1}
\begin{aligned}
u_i =cK\sum_{j=1}^Na_{ij}(x_i-x_j),\quad i=1,\cdots,N,
\end{aligned}
\end{equation}
where $c>0$ is the coupling weight among neighboring agents,
$K\in\mathbf{R}^{p\times n}$ is the feedback gain matrix, and
$a_{ij}$ is $(i,j)$-th entry of the adjacency matrix $\mathcal {A}$ associated with
$\mathcal {G}$, defined as $a_{ii}=0$, $a_{ij}=a_{ji}=1$ if
$(j,i)\in\mathcal {E}$ and $a_{ij}=a_{ji}=0$ otherwise. The
Laplacian matrix $\mathcal {L}=(\mathcal {L}_{ij})_{N\times N}$ of
$\mathcal {G}$ is defined by $\mathcal {L}_{ii}=\sum_{j=1,j\neq
i}^{N}a_{ij}$ and $\mathcal {L}_{ij}=-a_{ij}$ for $i\neq j$.

\begin{lemma}[\cite{li2010consensus}]\label{lem1}
Suppose that $\mathcal
{G}$ is connected. The $N$ agents described by \dref{1c} reach consensus
(i.e., $\lim_{t\rightarrow \infty}\|x_i(t)- x_j(t)\|=0$, $
\forall\,i,j=1,\cdots,N$)
under the protocol \dref{clc1} with $K=-B^TP^{-1}$ and the coupling
weight $c\geq \frac{1}{\lambda_2}$, where $\lambda_2$ is the
smallest nonzero eigenvalue of $\mathcal {L}$ and $P>0$ is a
solution to the following linear matrix inequality (LMI):
\begin{equation}\label{alg1}
AP+PA^T-2BB^T<0.
\end{equation}
\end{lemma}

As shown in the above lemma, the coupling weight $c$ should be not
less than the inverse of the smallest nonzero eigenvalue $\lambda_2$
of $\mathcal {L}$ to reach consensus. The design method for the
dynamic protocol in \cite{seo2009consensus} depends on $\lambda_2$
also. However, $\lambda_2$ is global information in the sense that
each agent has to know the Laplacian matrix and hence the entire
communication graph $\mathcal {G}$ to compute it.
Therefore, the consensus protocols given in Lemma \ref{lem1} and
\cite{seo2009consensus} cannot be implemented by each
agent in a fully distributed fashion, i.e., using only the local
information of its own and neighbors.

In order to avoid the limitation stated as above,
we propose the following distributed consensus protocol
with an adaptive law for adjusting the coupling weights:
\begin{equation}\label{clc2}
\begin{aligned}
u_i &
=F\sum_{j=1}^Nc_{ij}a_{ij}(x_i-x_j),\\
\dot{c}_{ij}& =\kappa_{ij}a_{ij}(x_i-x_j)^T\Gamma(x_i-x_j),\quad
i=1,\cdots,N,
\end{aligned}
\end{equation}
where $a_{ij}$ is defined as in
\dref{clc1}, $\kappa_{ij}=\kappa_{ji}$ are positive
constants, $c_{ij}$ denotes the time-varying coupling weight
between agents $i$ and $j$ with $c_{ij}(0)=c_{ji}(0)$,
and $F\in\mathbf{R}^{p\times n}$ and
$\Gamma\in\mathbf{R}^{n\times n}$ are feedback gain matrices.

We next design $F$ and $\Gamma$ in \dref{clc2} such that
the $N$ agents reach consensus.
%
%

\begin{theorem}\label{th1}
For any given connected graph $\mathcal
{G}$, the $N$ agents described by \dref{1c} reach consensus under the
protocol \dref{clc2} with $F=-B^TP^{-1}$ and
$\Gamma=P^{-1}BB^TP^{-1}$, where $P>0$ is a solution to the LMI
\dref{alg1}. Moreover, each coupling weight $c_{ij}$ converges to some finite
steady-state value.
\end{theorem}

\begin{proof}
Let $\bar{x}=\frac{1}{N}\sum_{j=1}^{N} x_j$, $e_i=x_i-\bar{x}$, and
$e=[e_1^T,\cdots,e_N^T]^T$. Then, we get
$e=\left((I_N-\frac{1}{N}\mathbf{1}\mathbf{1}^T)\otimes I_n\right)x.$ It is easy
to see that $0$ is a simple eigenvalue of
$I_N-\frac{1}{N}\mathbf{1}\mathbf{1}^T$ with $\mathbf{1}$ as the
corresponding right eigenvector, and 1 is the other eigenvalue with
multiplicity $N-1$. Then, it follows that $e=0$ if and only if
$x_1=\cdots=x_N$. Therefore, the consensus problem under the protocol
\dref{clc2} can be reduced to the
asymptotical stability of $e$. Using \dref{clc2} for \dref{1c}, it
follows that $e$ satisfies the following
dynamics:
\begin{equation}\label{netce1}
\begin{aligned}
\dot{e}_i &= Ae_i+\sum_{j=1}^Nc_{ij}a_{ij}BF(e_i-e_j),\\
\dot{c}_{ij}& =\kappa_{ij}a_{ij}(e_i-e_j)^T\Gamma(e_i-e_j),\quad
i=1,\cdots,N.
\end{aligned}
\end{equation}

Consider the Lyapunov function candidate
\begin{equation}\label{lya1}
V_1=\sum_{i=1}^{N}e_i^TP^{-1}e_i+ \sum_{i=1}^{N}\sum_{j=1,j\neq
i}^{N}\frac{(c_{ij}-\alpha)^2}{2\kappa_{ij}},
\end{equation}
where $\alpha$ is a positive constant. The time derivative of $V_1$
along the trajectory of \dref{netce1} is given by
\begin{equation}\label{lyae}
\begin{aligned}
\dot{V}_1
&=2\sum_{i=1}^{N}e_i^TP^{-1}\dot{e}_i+\sum_{i=1}^{N}\sum_{j=1,j\neq i}^{N}\frac{c_{ij}-\alpha}{\kappa_{ij}}\dot{c}_{ij}\\
&=2\sum_{i=1}^{N}e_i^TP^{-1}\left(Ae_i+\sum_{j=1}^Nc_{ij}a_{ij}BF(e_i-e_j)\right)\\
    &\quad+\sum_{i=1}^{N}\sum_{j=1,j\neq i}^{N}(c_{ij}-\alpha)a_{ij}(e_i-e_j)^T\Gamma(e_i-e_j).
\end{aligned}
\end{equation}

Because $\kappa_{ij}=\kappa_{ji}$, $c_{ij}(0)=c_{ji}(0)$, and $\Gamma$ is symmetric,
it follows from \dref{clc2} that
$c_{ij}(t)=c_{ji}(t)$, $\forall\,t\geq 0$. Therefore, we have
\begin{equation}\label{equa1}
\begin{aligned}
&\sum_{i=1}^{N}\sum_{j=1,j\neq i}^{N}(c_{ij}-\alpha)a_{ij}(e_i-e_j)^T\Gamma(e_i-e_j)\\
&\quad\qquad=2\sum_{i=1}^{N}\sum_{j=1}^{N}(c_{ij}-\alpha)a_{ij}e_i^T\Gamma
(e_i-e_j).
\end{aligned}
\end{equation}
Let $\tilde{e}_i=P^{-1}e_i$ and
$\tilde{e}=[\tilde{e}_1^T,\cdots,\tilde{e}_N^T]^T$.
Substituting $F=-B^TP^{-1}$ and
$\Gamma=P^{-1}BB^TP^{-1}$ into \dref{lyae}, we can
obtain
\begin{equation}\label{lyae2}
\begin{aligned}
\dot{V}_1
&=2\sum_{i=1}^{N}e_i^TP^{-1}Ae_i-2\alpha\sum_{i=1}^{N}\sum_{j=1}^Na_{ij}e_i^TP^{-1}BB^TP^{-1}(e_i-e_j)\\
&=\sum_{i=1}^{N}\tilde{e}_i^T(AP+PA^T)\tilde{e}_i-2\alpha\sum_{i=1}^{N}\sum_{j=1}^{N}\mathcal {L}_{ij}\tilde{e}_i^TBB^T\tilde{e}_j\\
&=\tilde{e}^T\left(I_N\otimes(AP+PA^T)-2\alpha\mathcal {L}\otimes
BB^T\right)\tilde{e},
\end{aligned}
\end{equation}
where $\mathcal {L}$ is the Laplacian matrix associated with
$\mathcal {G}$.

Because $\mathcal {G}$ is connected, zero is a simple eigenvalue of
$\mathcal {L}$ and all the other eigenvalues are positive
\cite{godsil2001algebraic}. Let
$U\in\mathbf{R}^{N\times N}$ be such a unitary matrix that
$U^{T}\mathcal
{L}U=\Lambda\triangleq{\rm{diag}}(0,\lambda_2,\cdots,\lambda_N)$.
Because the right and left eigenvectors of $\mathcal {L}$
corresponding to the zero eigenvalue are ${\bf 1}$ and ${\bf 1}^T$,
respectively, we can choose $U=\left[\begin{smallmatrix}
\frac{\mathbf{1}}{\sqrt{N}} & Y_1
\end{smallmatrix}\right]$ and $U^T=\left[\begin{smallmatrix}
\frac{\mathbf{1}^T}{\sqrt{N}} \\ Y_2
\end{smallmatrix}\right]$, with $Y_1\in\mathbf{R}^{N\times(N-1)}$ and $Y_2\in\mathbf{R}^{(N-1)\times N}$.
Let
$\xi\triangleq[\xi_1^T,\cdots,\xi_N^T]^T=(U^T\otimes I_n)\tilde{e}$.
By the definitions of $e$ and $\tilde{e}$, it is easy to see that
\begin{equation}\label{pcz1}
\xi_1=(\frac{\mathbf{1}^T}{\sqrt{N}}\otimes I_n)\tilde{e}=(\frac{\mathbf{1}^T}{\sqrt{N}}\otimes P^{-1})e=0.
\end{equation}
Then, we have
\begin{equation}\label{lyae3}
\begin{aligned}
\dot{V}_1 & = \xi^T\left(I_N\otimes(AP+PA^T)-2\alpha\Lambda\otimes BB^T\right)\xi\\
&=\sum_{i=2}^{N}\xi_i^T\left(AP+PA^T-2\alpha\lambda_iBB^T\right)\xi_i.
\end{aligned}
\end{equation}
By choosing $\alpha$ sufficiently large such that
$\alpha\lambda_i\geq1$, $i=2,\cdots,N$, it follows from \dref{alg1}
that $$AP+PA^T-2\alpha\lambda_iBB^T\leq AP+PA^T-2BB^T<0.$$
Therefore, $\dot{V}_1\leq 0$.

Since $\dot{V}_1\leq 0$, $V_1(t)$ is bounded and so is each $c_{ij}$.
By noting $\Gamma\geq 0$, it can be seen from \dref{netce1}
that $c_{ij}$ is monotonically increasing.
Then, it follows that each coupling weight $c_{ij}$ converges to some finite value.
Let $S=\{\xi_i,c_{ij}|\dot{V}_1=0\}$. Note that $\dot{V}_1\equiv0$
implies that $\xi_i=0$, $i=2,\cdots,N$, which, by noticing that
$\xi_1\equiv0$ in \dref{pcz1}, further implies that $\tilde{e}=0$
and $e=0$. Hence, by LaSalle's Invariance principle
\cite{krstic1995nonlinear}, it follows that $e(t)\rightarrow 0$, as
$t\rightarrow \infty$. That is, the consensus problem is solved.
\end{proof}

\begin{remark}
Equation \dref{clc2} presents an adaptive protocol,
under which the agents with dynamics given by \dref{1c} can reach consensus for all
connected communication topologies. In contrast to the consensus
protocols in \cite{li2010consensus,seo2009consensus}, the adaptive
protocol \dref{clc2} can be computed and implemented by each agent
in a fully distributed way. As shown in \cite{li2010consensus}, a
necessary and sufficient condition for the existence of a $P>0$ to
the LMI \dref{alg1} is that $(A,B)$ is stabilizable. Therefore, a
sufficient condition for the existence of a protocol \dref{clc2}
satisfying Theorem 1 is that $(A,B)$ is stabilizable. 
\end{remark}

\begin{remark}
It is worth noting that the consensus
protocols in
\cite{tuna2009conditions,scardovi2009synchronization} can also
achieve consensus for all connected communication graphs. Contrary
to the general linear agent dynamics in this section,
the agent dynamics in \cite{tuna2009conditions} are restricted to be
neutrally stable and all the eigenvalues of
the state matrix of each agent in \cite{scardovi2009synchronization}
are assumed to lie in the closed left-half plane. 
\end{remark}

\section{Adaptive Consensus for Multi-Agent Systems with Lipschitz Nonlinearity}

In this section, we study the consensus problem for a group of $N$
identical nonlinear agents, described by
\begin{equation}\label{lip}
\begin{aligned}
\dot{x}_i =Ax_i+D_1f(x_i)+Bu_i, \qquad i=1,\cdots,N,
\end{aligned}
\end{equation}
where $x_i\in\mathbf{R}^n$, $u_i\in\mathbf{R}^{p}$ are the state and
the control input of the $i$-th agent, respectively, $A$, $B$,
$D_1$, are constant matrices with compatible dimensions, and the
nonlinear function $f(x_i)$ is assumed to satisfy the Lipschitz
condition with a Lipschitz constant $\gamma>0$, i.e.,
\begin{equation}\label{lipcon}
\|f(x)-f(y)\|\leq \gamma\|x-y\|,\quad \forall\,x,y\in\mathbf{R}^n.
\end{equation}


\begin{theorem}
Solve the following LMI:
\begin{equation}\label{liplmi}
\begin{bmatrix} AQ+QA^T- \tau BB^T+\gamma^2 D_1T D_1^T & Q
    \\Q & -T\end{bmatrix}<0,
\end{equation}
to get a matrix $Q>0$, a scalar $\tau>0$ and a diagonal matrix
$T>0$. Then, the $N$ agents described by \dref{lip} reach
global consensus under the protocol \dref{clc2} with $F=-B^TQ^{-1}$
and $\Gamma=Q^{-1}BB^TQ^{-1}$ for any connected communication graph
$\mathcal {G}$.
Furthermore, each coupling weight $c_{ij}$ converges to some finite
steady-state value.
\end{theorem}

\begin{proof}
Using \dref{clc2} for \dref{lip}, we obtain the closed-loop network
dynamics as
\begin{equation}\label{netlip}
\begin{aligned}
\dot{x}_i &=Ax_i+D_1f(x_i)+\sum_{j=1}^{N}c_{ij}a_{ij}BF(x_i-x_j),\\
\dot{c}_{ij}& =\kappa_{ij}a_{ij}(x_i-x_j)^T\Gamma(x_i-x_j),\quad
i=1,\cdots,N.
\end{aligned}
\end{equation}
As argued in the Proof of Theorem 1, it follows that
$c_{ij}(t)=c_{ji}(t)$, $\forall\,t\geq 0$.

Letting $\bar{x}=\frac{1}{N}\sum_{j=1}^{N} x_j$,
$e_i=x_i-\bar{x}$, and $e=[e_1^T,\cdots,e_N^T]^T$, we get
$e=\left((I_N-\frac{1}{N}\mathbf{1}\mathbf{1}^T)\otimes
I_n\right)x.$ By following similar steps to those in Theorem 1, we
can reduce the consensus problem of \dref{netlip} to the
convergence of $e$ to the origin. It is easy to obtain that $e$ satisfies the following
dynamics:
\begin{equation}\label{netlip2}
\begin{aligned}
\dot{e}_i
&=Ae_i+D_1f(x_i)-\frac{1}{N}\sum_{j=1}^{N}D_1f(x_j)+\sum_{j=1}^{N}(\tilde{c}_{ij}+\beta)a_{ij}BF(e_i-e_j),\\
\dot{\tilde{c}}_{ij}& =\kappa_{ij}a_{ij}(e_i-e_j)^T\Gamma(e_i-e_j),\quad
i=1,\cdots,N,
\end{aligned}
\end{equation}
where $c_{ij}=\tilde{c}_{ij}+\beta$ and $\beta$ is a positive constant.

Consider the Lyapunov function candidate
\begin{equation}\label{lyalip}
V_2=\sum_{i=1}^{N}e_i^TQ^{-1}e_i+ \sum_{i=1}^{N}\sum_{j=1,j\neq
i}^{N}\frac{\tilde{c}_{ij}^2}{2\kappa_{ij}}.
\end{equation}
The time derivative of $V_2$
along the trajectory of \dref{netlip2} is
\begin{equation}\label{lyalip2}
\begin{aligned}
\dot{V}_2
&=2\sum_{i=1}^{N}e_i^TQ^{-1}\dot{e}_i+\sum_{i=1}^{N}\sum_{j=1,j\neq
i}^{N}\frac{\tilde{c}_{ij}}{\kappa_{ij}}\dot{\tilde{c}}_{ij}\\
&=2\sum_{i=1}^{N}e_i^TQ^{-1}\left(Ae_i+D_1f(x_i)-\frac{1}{N}\sum_{j=1}^{N}D_1f(x_j)+\sum_{j=1}^N(\tilde{c}_{ij}+\beta)a_{ij}BF(e_i-e_j)\right)\\
    &\quad+\sum_{i=1}^{N}\sum_{j=1,j\neq
    i}^{N}\tilde{c}_{ij}a_{ij}(e_i-e_j)^T\Gamma(e_i-e_j)\\
&=2\sum_{i=1}^{N}e_i^TQ^{-1}Ae_i-2\beta\sum_{i=1}^{N}\sum_{j=1}^N\mathcal {L}_{ij}e_i^TQ^{-1}BB^TQ^{-1}e_j\\
&\quad +2\sum_{i=1}^{N}e_i^TQ^{-1}D_1\left(
f(x_i)-f(\bar{x})+f(\bar{x})-\frac{1}{N}\sum_{j=1}^{N}f(x_j)\right),
\end{aligned}
\end{equation}
where we have used the fact \dref{equa1} to get the
last equation.

Using the Lipschitz condition \dref{lipcon} gives
\begin{equation}\label{lipc2}
\begin{aligned}
2e_i^TQ^{-1}D_1(f(x_i)-f(\bar{x}))&\leq
2\gamma\|e_i^TQ^{-1}D_1T^{\frac{1}{2}}\|\|T^{-\frac{1}{2}}e_i\|
\\&\leq e_i^T(\gamma^2Q^{-1}D_1T D_1^TQ^{-1}+T^{-1})e_i,
\end{aligned}
\end{equation}
where $T$ is given in \dref{liplmi}. Because
$\sum_{i=1}^Ne_i^T=0$, we have
\begin{equation}\label{lipc3}
\sum_{i=1}^Ne_i^TQ^{-1}D_1\left(f(\bar{x})-\frac{1}{N}\sum_{j=1}^{N}f(x_j)\right)=0.
\end{equation}
Let $\hat{e}_i=Q^{-1}e_i$ and
$\hat{e}=[\hat{e}_1^T,\cdots,\hat{e}_N^T]^T$. In virtue of
\dref{lipc2} and \dref{lipc3}, we can obtain from \dref{lyalip2} that
\begin{equation}\label{lyalip3}
\begin{aligned}
\dot{V}_2
&\leq\sum_{i=1}^{N}\hat{e}_i^T\left((AQ+QA^T+\gamma^2D_1T
D_1^T+QT^{-1}Q)\hat{e}_i-2\beta\sum_{j=1}^{N}\mathcal
{L}_{ij}BB^T\hat{e}_j\right)\\
&=\hat{e}^T\left(I_N\otimes (AQ+QA^T+\gamma^2D_1T
D_1^T+QT^{-1}Q)-2\beta\mathcal {L}\otimes BB^T\right)\hat{e}.
\end{aligned}
\end{equation}
Let $U\in\mathbf{R}^{N\times N}$ be the unitary matrix defined in
the proof of Theorem 1, satisfying $U^{T}\mathcal
{L}U=\Lambda={\rm{diag}}(0,\lambda_2,\cdots,\lambda_N)$. Let
$\zeta\triangleq[\zeta_1^T,\cdots,\zeta_N^T]^T=(U^T\otimes
I_n)\hat{e}$. Clearly,
$\zeta_1=(\frac{\mathbf{1}^T}{\sqrt{N}}\otimes Q^{-1})e=0$. From
\dref{lyalip3}, we have
\begin{equation}\label{lyalip4}
\begin{aligned}
\dot{V}_2
&\leq \zeta^T\left(I_N\otimes (AQ+QA^T+\gamma^2D_1T D_1^T+QT^{-1}Q)-2\beta\Lambda\otimes BB^T\right)\zeta\\
&=\sum_{i=2}^{N}\zeta_i^T\left(AQ+QA^T+\gamma^2D_1T D_1^T+QT^{-1}Q-2\beta\lambda_iBB\right)\zeta_i\\
&\triangleq W(\zeta).
\end{aligned}
\end{equation}
By choosing $\beta$ sufficiently large such that
$2\beta\lambda_i\geq\tau$, $i=2,\cdots,N$, it follows that
$$
\begin{aligned}
&AQ+QA^T-2\beta\lambda_i BB^T+\gamma^2D_1T D_1^T+QT^{-1}Q\\
&\qquad \qquad\leq AQ+QA^T-\tau BB^T+\gamma^2D_1T
D_1^T+QT^{-1}Q\\
&\qquad \qquad<0,\quad i=2,\cdots,N,
\end{aligned}
$$
where the last inequality follows from \dref{liplmi} by using the
Schur complement lemma \cite{boyd1994linear}. Therefore,
$W(\zeta)\leq0$.

Since $\dot{V}_2\leq 0$, $V_2(t)$ is bounded and so is each $\tilde{c}_{ij}$.
By \dref{netlip2}, $\tilde{c}_{ij}$ is monotonically increasing. Then, it follows that each
$\tilde{c}_{ij}$ converges to some finite value. Thus the coupling weights $c_{ij}$
converge to finite steady-state values.
Note that $V_2$ is positive definite and radically unbounded.
By LaSalle-Yoshizawa theorem \cite{krstic1995nonlinear}, it
follows that $\lim_{t\rightarrow \infty}W(\zeta)=0$, implying
that $\zeta_i\rightarrow 0$, $i=2,\cdots,N$, as $t\rightarrow \infty$,
which, together with $\zeta_1\equiv0$, further
implies that $e(t)\rightarrow 0$, as $t\rightarrow \infty$. This completes the proof.
\end{proof}

\begin{remark}
By using Finsler's Lemma
\cite{iwasaki1994all}, it is not difficult to see that there exist a
$Q>0$, a $T>0$ and a $\tau>0$ such that \dref{liplmi} holds if
and only if there exists a $K$ such that
$(A+BK)Q+Q(A+BK)^T+\gamma^2D_1T D_1^T+QT^{-1}Q<0$, which
with $T=I$ is dual to the observer design problem for a single
Lipschitz system in \cite{rajamani1998observers,rajamani1998existence}.
According to Theorem 2 in \cite{rajamani1998existence},
the LMI \dref{liplmi} is feasible, and thus there exists an adaptive protocol
\dref{clc2} reaching consensus, if the distance to unobservability of
$(A,B)$ is larger than $\gamma$. Besides, a
diagonal scaling matrix $T>0$ is introduced here in
\dref{liplmi} to reduce conservatism. If the nonlinear function
$f(x_i)=0$ in \dref{lip}, then \dref{lip} becomes \dref{1c}. By
choosing $T$ sufficiently large and letting $D_1=0$ and
$\tau=2$, then \dref{liplmi} becomes $AQ+QA^T-2 BB^T<0$. Therefore,
for the case without the Lipschitz nonlinearity, Theorem 2 is
reduced to Theorem 1.
\end{remark}

\begin{remark}
It should be mentioned that the adaptive law
in \dref{clc2} for adjusting the coupling weights is inspired by the
edge-based adaptive strategy in \cite{delellis2009novel,delellis2010synchronization},
where adaptive synchronization of a class of complex network satisfying a
Lipschitz-type condition is considered.
However, the results given in \cite{delellis2009novel,delellis2010synchronization} require the
inner coupling matrix to be positive semi-definite, and are thereby not
directly applicable to the consensus problem under
investigation here.
\end{remark}

\section{Extensions}

The communication topology is assumed to be undirected in the previous
sections, where the final consensus value reached by the agents is
generally not explicitly known, due to the nonlinearity in
the closed-loop network dynamics. In many practical cases, it is
desirable for the agents' states to asymptotically approach a
reference state. In this section, we consider the case where a
network of $N+1$ agents maintains a leader-follower communication
structure.

The agents' dynamics remain the same as in \dref{1c}. The agents
indexed by $1,\cdots,N$, are referred to as followers, while the
agent indexed by 0 is called the virtual leader whose control input
$u_0=0$. The communication topology among the $N$ followers is
represented by an undirected graph $\mathcal {G}$. It is assumed
that the leader receives no information from any follower and the
state of the leader is available to only a subset of the followers
(without loss of generality, the first $q$ followers). In this case,
the following distributed consensus protocol is proposed
\begin{equation}\label{clf}
\begin{aligned}
u_i &
=\hat{F}\left(\sum_{j=1}^Nc_{ij}a_{ij}(x_i-x_j)+c_id_i(x_i-x_0)\right),\\
\dot{c}_{ij}& =\kappa_{ij}a_{ij}(x_i-x_j)^T\hat{\Gamma}(x_i-x_j),\\
\dot{c}_{i}& =\kappa_{i}d_i(x_i-x_0)^T\hat{\Gamma}(x_i-x_0),\quad
i=1,\cdots,N,
\end{aligned}
\end{equation}
where $a_{ij}$, $c_{ij}$, $\kappa_{ij}$ are defined as in
\dref{clc2}, $c_{i}$ denotes the coupling weight between agent $i$
and the virtual leader, $\kappa_{i}$ are positive constants,
$\hat{F}\in\mathbf{R}^{p\times n}$ and
$\hat{\Gamma}\in\mathbf{R}^{n\times n}$ are feedback gain matrices,
and $d_i$ are constant gains, satisfying $d_i>0$, $i=1,\cdots,q,$ and $d_i=0$,
$i=q+1,\cdots,N$.

The objective here is to design $\hat{F}\in\mathbf{R}^{p\times n}$
and $\hat{\Gamma}\in\mathbf{R}^{n\times n}$ such that the states of
the followers can asymptotically approach the state of the leader in
the sense that $\lim_{t\rightarrow \infty}\|x_i(t)- x_0(t)\|=0$,
$\forall\,i=1,\cdots,N.$

\begin{theorem}
Assume that $\mathcal {G}$ is connected and
at least one follower can have access to the leader's state. Then,
the states of the $N$ followers asymptotically approach the state of
the leader, under the protocol \dref{clf} with $\hat{F}=-B^TP^{-1}$
and $\hat{\Gamma}=P^{-1}BB^TP^{-1}$, where $P>0$ is a solution to
\dref{alg1}, and the coupling weights $c_{ij}$ and $c_{i}$ converge
to finite values.
\end{theorem}

\begin{proof}
 Let $\varepsilon_i=x_i-x_0$, $i=1,\cdots,N$.
Then, the collective network dynamics resulting from \dref{1c} and
\dref{clf} can be written as
\begin{equation}\label{netcd1f}
\begin{aligned}
\dot{\varepsilon}_i &= A\varepsilon_i+B\hat{F}\left(\sum_{j=1}^Nc_{ij}a_{ij}(\varepsilon_i-\varepsilon_j)+c_id_i\varepsilon_i\right),\\
\dot{c}_{ij}&
=\kappa_{ij}a_{ij}(\varepsilon_i-\varepsilon_j)^T\hat{\Gamma}(\varepsilon_i-\varepsilon_j),\\
\dot{c}_{i}&
=\kappa_{i}d_i\varepsilon_i^T\hat{\Gamma}\varepsilon_i,\quad
i=1,\cdots,N.
\end{aligned}
\end{equation}
Clearly, the states of the followers under \dref{clf} can
asymptotically approach the state of the leader, if \dref{netcd1f}
is asymptotically stable.

Consider the Lyapunov function candidate
\begin{equation}\label{lyaf1}
V_3=\sum_{i=1}^{N}\varepsilon_i^TP^{-1}\varepsilon_i+
\sum_{i=1}^{N}\sum_{j=1,j\neq
i}^{N}\frac{(c_{ij}-\beta)^2}{2\kappa_{ij}}+\sum_{i=1}^{N}\frac{(c_{i}-\beta)^2}{\kappa_{i}},
\end{equation}
where $\beta$ is a positive constant. The rest of the proof follows
similar steps to those in Theorem 1, and by further noting
the fact: Suppose that $R=\rm{diag}(d_1,d_2,\cdots,d_N)\geq 0$ with
at least one diagonal item being positive. Then, $\mathcal {L}+R$ is
positive definite if $\mathcal {G}$ is connected
\cite{hong2008distributed}.
\end{proof}

\begin{remark}
It is worth mentioning that an adaptive pinning scheme
similar to \dref{clf} has been proposed in \cite{delellis2010fully}.
Compared to \cite{delellis2010fully} where
the inner coupling matrix is an identity matrix,
the adaptive protocol \dref{clf} here is more general.
\end{remark}

The case with the agents described by \dref{lip} can be discussed
similarly, and is thus omitted here for brevity.

\section{Simulation Examples}

In this section, a simulation example is provided to validate the
effectiveness of the theoretical results.

\begin{figure}[htbp]
\centering
\includegraphics[width=0.28\linewidth]{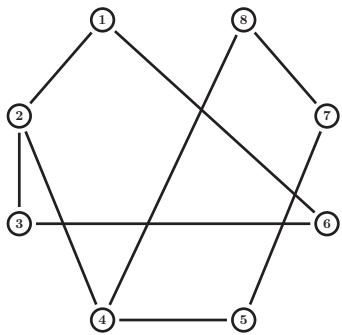}
\caption{The communication topology. }
\end{figure}

Consider a network of single-link manipulators with revolute joints
actuated by a DC motor. The dynamics of the $i$-th manipulator is
described by \dref{lip}, with (see \cite{rajamani1998existence})
$$\begin{aligned}
x_{i} &=\begin{bmatrix}
x_{i1}\\x_{i2}\\x_{i3}\\x_{i4}\end{bmatrix},\quad
A=\begin{bmatrix}0 & 1& 0 &0\\ -48.6 & -1.25 & 48.6 & 0\\
0 &0 & 0 & 10\\ 1.95 & 0 & -1.95 &0\end{bmatrix},\quad B
=\begin{bmatrix} 0 \\ 21.6 \\ 0 \\ 0\end{bmatrix},
\\D_1&=I_4, \quad f(x_i)=\begin{bmatrix} 0 & 0& 0&
-0.333\mathrm{sin}(x_{i3})\end{bmatrix}^T.
\end{aligned}$$
Clearly, $f(x_i)$ here satisfies \dref{lipcon} with a Lipschitz
constant $\gamma=0.333$.

Solving the LMI \dref{liplmi} by using the LMI toolbox of Matlab
gives the feedback gain matrices in \dref{clc2} as
$$\begin{aligned}
F &=\begin{bmatrix} -1.8351 &  -0.2144  &  1.0309 &  -2.247\end{bmatrix},\\
\Gamma &=\left[\begin{matrix} 3.3676  &  0.3935 &  -1.8917 &
4.1236\\ 0.3935  &  0.046  & -0.221 &   0.4818\\
-1.8917 &  -0.221  &  1.0627  & -2.3164\\
4.1236  &  0.4818  & -2.3164  &  5.0492\end{matrix}\right].
\end{aligned}$$
To illustrate Theorem 2, let the communication graph $\mathcal {G}$
be given in Fig. 1. Here $\mathcal {G}$ is undirected and connected.
Let $\kappa_{ij}=1$, $i,j=1,\cdots,8$, in \dref{clc2}, and
$c_{ij}(0)=c_{ji}(0)$ be randomly chosen. The states trajectories of
the eight manipulators under the protocol \dref{clc2} are depicted
in Fig. 2, from which it can be observed that consensus is
reached. The coupling weights $c_{ij}$ are shown in Fig. 3, which
converge to finite steady-state values.

\begin{figure}[htbp]\centering
\centering
\includegraphics[height=0.25\linewidth,width=0.4\linewidth]{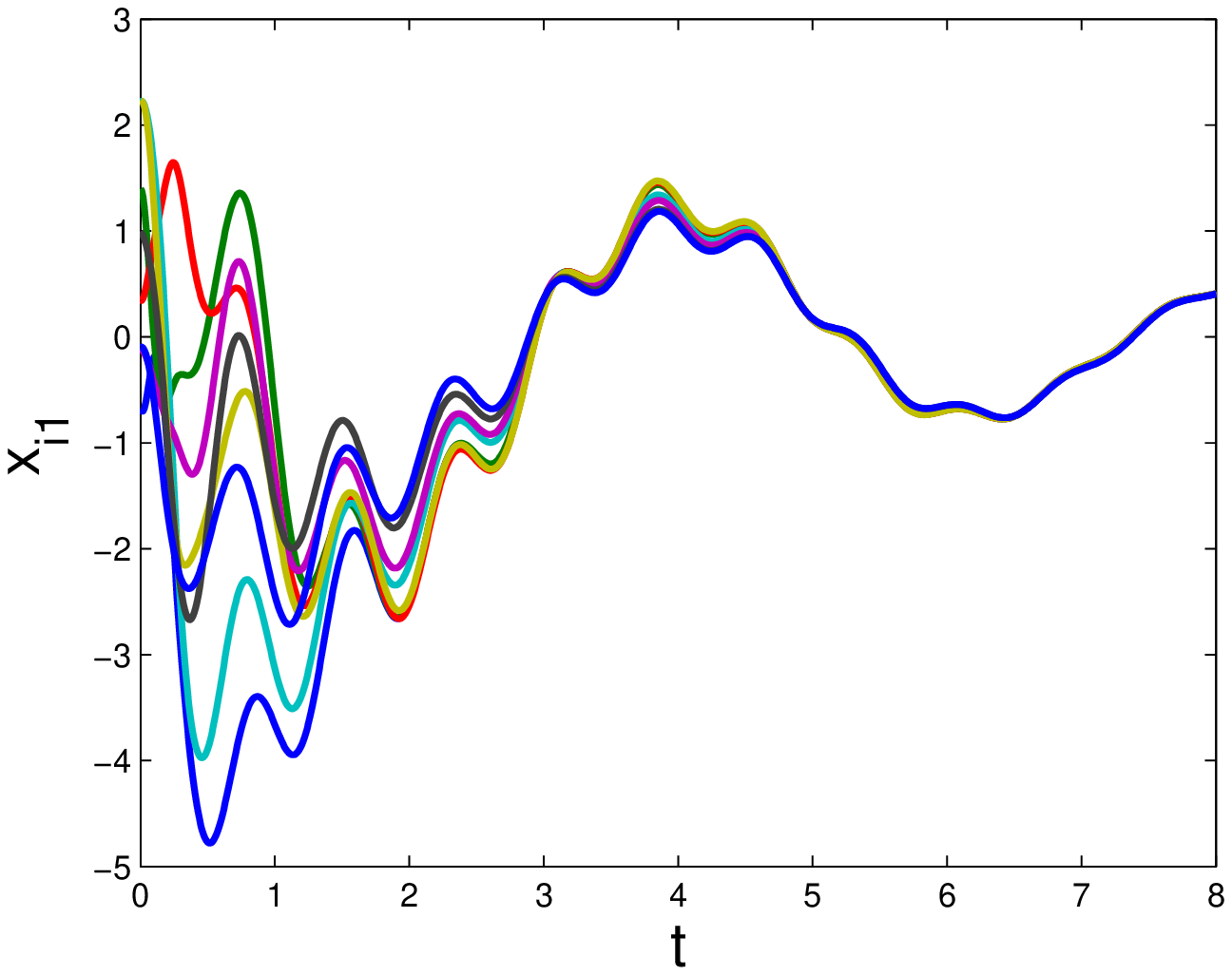}\quad
\includegraphics[height=0.25\linewidth,width=0.4\linewidth]{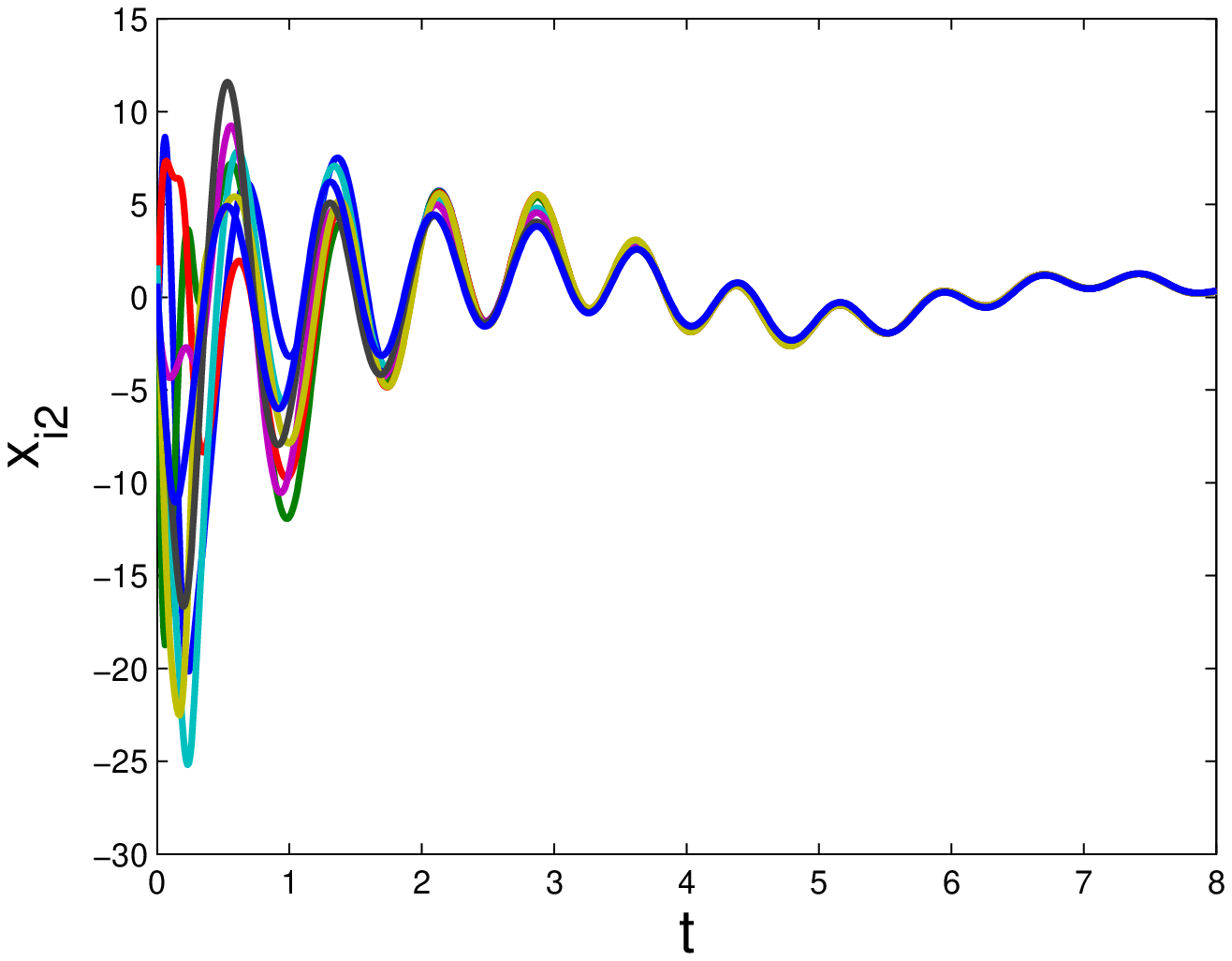}\\
\includegraphics[height=0.25\linewidth,width=0.4\linewidth]{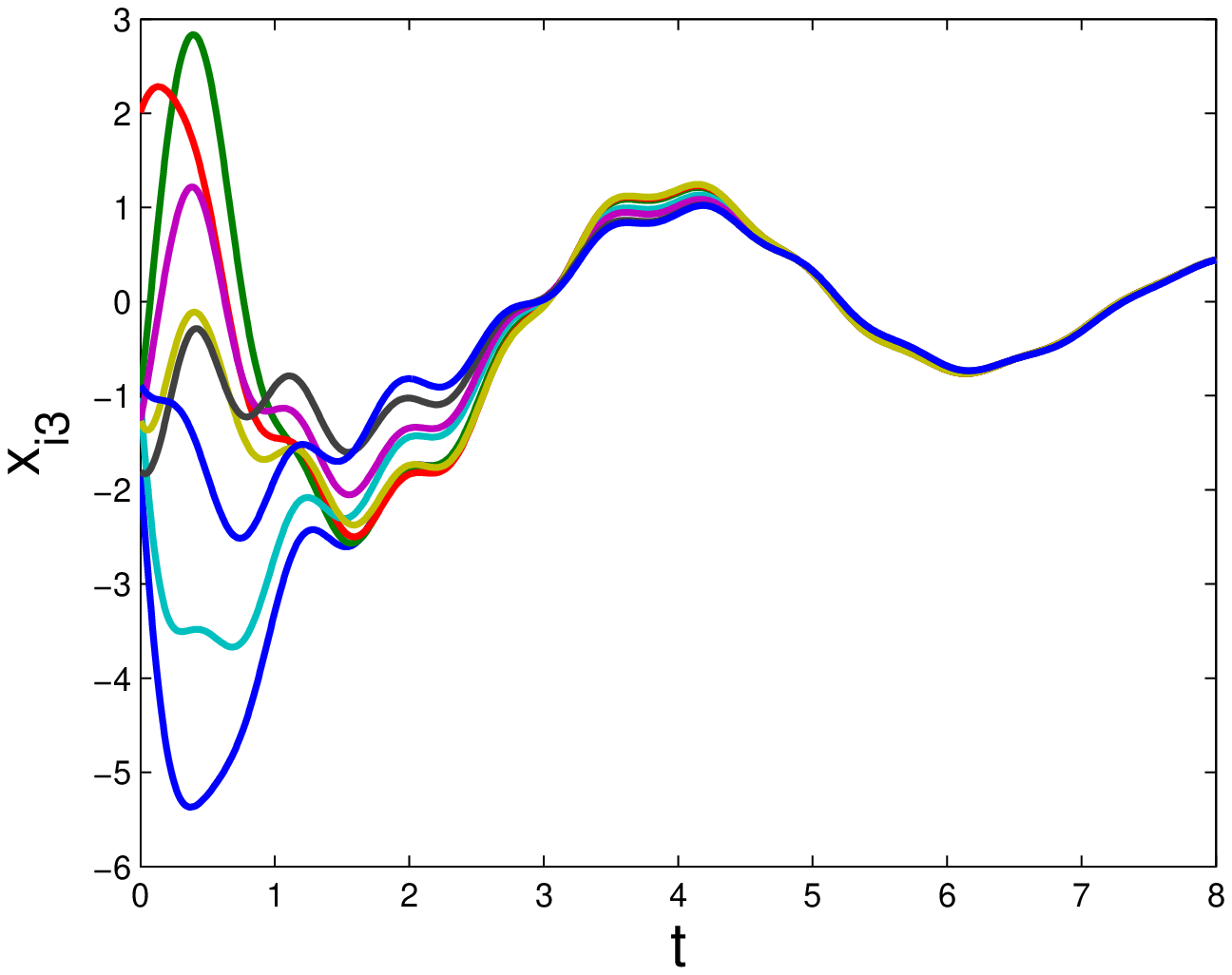}\quad
\includegraphics[height=0.25\linewidth,width=0.4\linewidth]{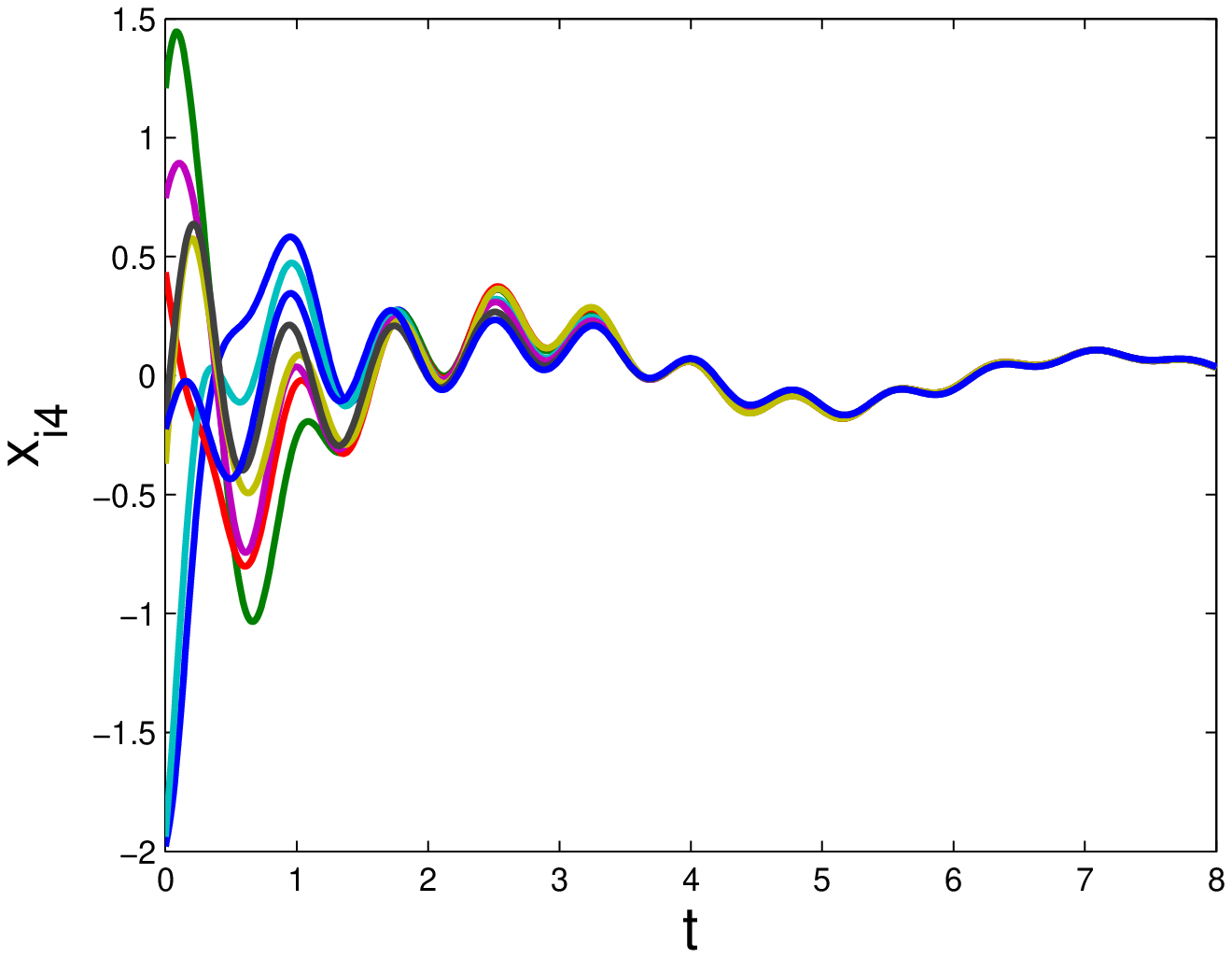}
\caption{The states of the eight manipulators under \dref{clc2}.}
\end{figure}

\begin{figure}[htbp]\centering
\centering
\includegraphics[height=0.28\linewidth,width=0.5\linewidth]{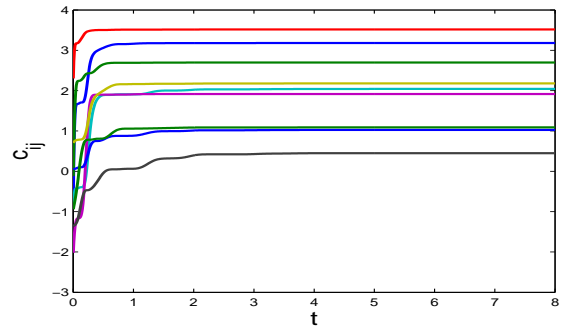}
\caption{The coupling weights $c_{ij}$.}
\end{figure}

\section{Conclusion}

In this paper, the distributed consensus problems have been considered
for multi-agent systems with general linear and
Lipschitz nonlinear dynamics.
Distributed relative-state consensus protocols with an adaptive
law for adjusting the coupling weights between neighboring agents
are designed for both the linear and nonlinear cases,
under which consensus is reached for all
undirected connected communication graphs. Extensions
to the case with a leader-follower communication graph have also been studied.

%
%

%
%

\end{document}